\documentclass[12pt,letterpaper]{article}
\usepackage{amsmath,amsfonts,amsthm,amssymb,tabls}
\usepackage{graphicx,relsize,cite}

\theoremstyle{plain}

\numberwithin{equation}{section}
\newtheorem{thm}{Theorem}[section]
\newtheorem{lem}[thm]{Lemma}
\newtheorem{cor}[thm]{Corollary}

\newenvironment{exam}[1]%
{\begin{flushleft}\textbf{Example #1}.\enspace}%
{\end{flushleft}}

\newcommand{\complex}{{\mathbb C}}
\newcommand{\positive}{{\mathbb N}}
\newcommand{\real}{{\mathbb R}}
\newcommand{\ascript}{{\mathcal A}}
\newcommand{\cscript}{{\mathcal C}}
\newcommand{\dscript}{{\mathcal D}}
\newcommand{\pscript}{{\mathcal P}}

\newcommand{\rscript}{{\mathcal R}}
\newcommand{\sscript}{{\mathcal S}}
\newcommand{\tscript}{{\mathcal T}}
\newcommand{\itim}{\mathop{\mathit{Im}}}
\newcommand{\itre}{\mathop{\mathit{Re}}}

\newcommand{\rmcyl}{\mathop{\mathrm{cyl}}}
\newcommand{\dscripthat}{\widehat{\dscript}}
\newcommand{\rscripthat}{\widehat{\rscript}}
\newcommand{\tscripthat}{\widehat{\tscript}}
\newcommand{\dbar}{\overline{D}}
\newcommand{\atilde}{\widetilde{a}}

\newcommand{\ab}[1]{\left|#1\right|}

\newcommand{\brac}[1]{\left\{#1\right\}}
\newcommand{\paren}[1]{\left(#1\right)}
\newcommand{\sqbrac}[1]{\left[#1\right]}
\newcommand{\elbows}[1]{{\left\langle#1\right\rangle}}
\newcommand{\ket}[1]{{\left|#1\right>}}
\newcommand{\bra}[1]{{\left<#1\right|}}

\begin{document}

\title{AN APPROACH TO\\DISCRETE QUANTUM GRAVITY
}
\author{S. Gudder\\ Department of Mathematics\\
University of Denver\\ Denver, Colorado 80208, U.S.A.\\
sgudder@du.edu
}
\date{}
\maketitle

\begin{abstract}
This article presents a simplified version of the author's previous work. We first construct a causal growth process (CGP). We then form path Hilbert spaces using paths of varying lengths in the CGP. A sequence of positive operators on these Hilbert spaces that satisfy certain normalization and consistency conditions is called a quantum sequential growth process (QSGP). The operators of a QSGP are employed to define natural decoherence functionals and quantum measures. These quantum measures are extended to a single quantum measure defined on a suitable collection of subsets of a space of all paths. Continuing our general formalism, we define curvature operators and a discrete analogue of Einstein's field equations on the Hilbert space of causal sets. We next present a method for constructing a QSGP using an amplitude process (AP). We then consider a specific AP that employs a discrete analogue of a quantum action. Finally, we consider the special case in which the QSGP is classical. It is pointed out that this formalism not only gives a discrete version of general relativity, there is also emerging a discrete analogue of quantum field theory. We therefore have discrete versions of these two theories within one unifying framework.
\end{abstract}

\section{Introduction}  
In some previous articles, the author developed a model for discrete quantum gravity by first constructing a classical sequential growth process (CSGP) \cite{gud111,gud112,gud13,hen09,rs00,sor94,sor03,sur11,vr06}. Roughly speaking, the CSGP corresponded to the classical configuration space for a system of multi-universes. As is frequently done in quantum theory, we then quantized the CSGP to form a Hilbert space. The role of a quantum state was played by a sequence of probability operators on an increasing sequence of Hilbert spaces. In the present article we give a simplified version in which we dispense with the CSGP and immediately begin with a quantum sequential grown process (QSGP). This approach appears to be cleaner and more direct.

We begin by constructing a causet growth process (CGP) for the collection of causal sets (causets) $\pscript$
\cite{rs00,sor03,sur11}. We then define the space of paths $\Omega$ and the space $\Omega _n$ of $n$-paths of length $n$ in $\pscript$. Letting $\ascript$ be the $\sigma$-algebra of subsets of $\Omega$ generated by the cylinder sets
$\cscript (\Omega )$, we obtain the measurable space $(\Omega ,\ascript )$. For a set $A\in\ascript$, its $n$-th approximation
$A^n\subseteq\Omega _n$ is the set of $n$-paths that can be extended to a path in $\Omega$. Forming the Hilbert spaces $H_n=L_2(\Omega _n)$, a sequence of positive operators $\rho _n$ on $H_n$ satisfying a normalization and consistency condition is called a quantum sequential growth process (QSGP) \cite{gud13}. The probability operators $\rho _n$ are employed to define natural decoherence functionals $D_n(A,B)$ and quantum measures \cite{sor94} $\mu _n(A)=D_n(A,A)$ on $\Omega _n$. A set $A\in\ascript$ is called \textit{suitable} if $\lim\mu _n(A^n)$ exists and in this case we define $\mu (A)$ to be this limit. Denoting the collection of suitable sets by $\sscript (\Omega )$ we have that
$\cscript (\Omega )\subseteq\sscript (\Omega )\subseteq\ascript$ and, in general, the inclusions are proper. In a certain sense, $\mu$ becomes a quantum measure on $\sscript (\Omega )$ so $\paren{\Omega ,\sscript (\Omega ),\mu}$ becomes a quantum measure space \cite{sor94}. Continuing our general formalism, we present a discrete analogue to Einstein's field equations. This is accomplished by defining curvature operators $\rscript _{\omega ,\omega '}$ on $L_2(\pscript )\otimes L_2(\pscript )$ for every $\omega ,\omega '\in\Omega$.

We next present a method for constructing a QSGP using an amplitude process (AP). An AP is essentially a Markov chain with complex-valued transition amplitudes. We show that an AP simplifies our previous formalism. We then consider a specific AP that employs a discrete analogue of a quantum action. We use this AP to illustrate some of our previous theory. Finally, we consider the special case in which the QSGP is classical. In this case the decoherence matrices $D_n(\omega ,\omega ')$ become diagonal and the quantum measure $\mu _n$ becomes ordinary measures. We end by studying the possibility that the curvature operator vanishes in this case.

\section{The Causet Growth Process} 
In this article we call a finite partially ordered set a \textit{causet}. Two isomorphic causets are considered to be identical. Let
$\pscript _n$ be the collection of al causets of cardinality $n$, $n=1,2,\ldots$, and let $\pscript =\cup\pscript _n$ be the collection of all causets. If $a,b$ are elements of a causet $x$, we interpret the order $a<b$ as meaning that $b$ is in the causal future of $a$ and $a$ is in the causal past of $b$. An element $a\in x$, for $x\in\pscript$ is \textit{maximal} if there is no $b\in x$ with $a<b$. If $x\in\pscript _n$, $y\in\pscript _{n+1}$, then $x$ \textit{produces} $y$ if $y$ is obtained from $x$ by adjoining a single maximal element $a$ to $x$. In this case we write $y=x\uparrow a$ and use the notation $x\to y$. If $x\to y$, we also say that $x$ is a \textit{producer} of $y$ and $y$ is an \textit{offspring} of $x$. Of course, $x$ may produce many offspring and $y$ may be the offspring of many producers. Moreover, $x$ may produce $y$ in various isomorphic ways \cite{rs00,sor03,vr06}. We denote by $m(x\to y)$ the number of isomorphic ways that $x$ produces $y$ and call $m(x\to y)$ the \textit{multiplicity} of  $x\to y$.

If $a,b\in x$ with $x\in\pscript$, we say that $a$ and $b$ are \textit{comparable} if $a\le b$ or $b\le a$. A \textit{chain} in $x$ is a set of mutually comparable elements of $x$ and an \textit{antichain} is a set of mutually incomparable elements of $x$. By convention, the empty set in both a chain and an antichain and clearly singleton sets also have this property. A chain is
\textit{maximal} if it is not a proper subset of a larger chain. The following result was proved in \cite{gud13}. We denote the cardinality of a finite set $A$ by $\ab{A}$. 

\begin{thm}       
\label{thm21}
The number of offspring $r$ of a causet $x$, including multiplicity, is the number of distinct antichains in $x$. We have that
$\ab{x}+1\le r\le 2^{\ab{x}}$ with both bounds achieved.
\end{thm}

The transitive closure of $\,\rightarrow\,$ makes $\pscript$ into a poset itself and we call $(\pscript,\rightarrow )$ the causet growth process (CGP) \cite{hen09,rs00,sor03,sur11}. Figure~1 illustrates the first three steps of the CGP. The 2 in Figure~1 denotes the fact that $m(x_2\to x_6)=2$. It follows from Theorem~\ref{thm21} that the number of offspring for $x_4$, $x_5$, $x_6$, $x_7$ and $x_8$ are $4,5,6,5,8$, respectively. In this case, Theorem~\ref{thm21} tells us that $4\le r\le 8$.

\setlength{\unitlength}{8pt}
\begin{picture}(0,20)
\put(24,8){\circle{7}}  
\put(24,8){\circle*{.35}}
\put(19.5,7){\makebox{$x_1$}}
\put(21.5,8){\vector(-1,1){4}} 
\put(26.5,8){\vector(1,1){4}} 
\put(16,14){\circle{7}} 
\put(14,15.5){\vector(-1,1){5.5}} 
\put(18,15,5){\vector(1,1){5}} 
\put(16,15){\circle*{.35}}
\put(16,13){\circle*{.35}}
\put(16,13.2){\line(0,1){1.5}}
\put(11.75,13.75){\makebox{$x_2$}}
\put(16,16.5){\vector(0,1){4}} 
\put(32,14){\circle{7}} 
\put(30,15.5){\vector(-1,1){5}} 
\put(34,15.5){\vector(1,1){5}} 
\put(27,17){\makebox{$\scriptstyle 2$}}
\put(31,14){\circle*{.35}}
\put(33,14){\circle*{.35}}
\put(35,13.5){\makebox{$x_3$}}
\put(32,16.5){\vector(0,1){4}} 
\put(7,23){\circle{7}} 
\put(7,21.5){\circle*{.35}}
\put(7,23){\circle*{.35}}
\put(7,24.5){\circle*{.35}}
\put(7,21.5){\line(0,1){3}}
\put(2.75,23){\makebox{$x_4$}}
\put(16,23){\circle{7}} 
\put(16,22){\circle*{.35}}
\put(15.25,23.5){\circle*{.35}}
\put(16.75,23.5){\circle*{.35}}
\put(16,22){\line(1,2){.8}}
\put(16,22){\line(-1,2){.8}}
\put(12,23){\makebox{$x_5$}}
\put(24,23){\circle{7}} 
\put(23,22){\circle*{.35}}
\put(24.5,22){\circle*{.35}}
\put(23,23.5){\circle*{.35}}
\put(23,22){\line(0,1){1.5}}
\put(20,23){\makebox{$x_6$}}
\put(32,23){\circle{7}} 
\put(31,22){\circle*{.35}}
\put(33,22){\circle*{.35}}
\put(32,23.75){\circle*{.35}}
\put(31,22){\line(1,2){.8}}
\put(33,22){\line(-1,2){.8}}
\put(28,23){\makebox{$x_7$}}
\put(40,23){\circle{7}} 
\put(38.5,23){\circle*{.35}}
\put(40,23){\circle*{.35}}
\put(41.5,23){\circle*{.35}}
\put(43,23){\makebox{$x_8$}}
\centerline{\textbf{Figure 1}}
\end{picture}
\bigskip

Theorem~\ref{thm21} has a kind of dual theorem that we now discuss. Two maximal chains $C_1,C_2$ in a causet $x$ are \textit{equivalent} if, when the maximal elements $a_1$ and $a_2$ of $C_1$, $C_2$, respectively, are deleted, then the resulting causets $C_1\smallsetminus\brac{a_1}$, $C_2\smallsetminus\brac{a_2}$ are isomorphic. For example, in Figure~1, the maximal chains in $x_4,x_5,x_7,x_8$ are mutually equivalent while $x_6$ has two inequivalent maximal chains. Notice that $x_4,x_5,x_7,x_8$ each have one producer, while $x_6$ has two producers. This observation motivates
Theorem~\ref{thm23}. But first we prove a lemma.

\begin{lem}       
\label{lem22}
The definition of equivalent maximal chains gives an equivalence relation.
\end{lem}
\begin{proof}
If $C_1,C_2$ are equivalent maximal chains in a causet $x$, we write $C_1\sim C_2$. It is clear that $C_1\sim C_1$ and $C_1\sim C_2$ implies $C_2\sim C_1$. To prove transitivity, suppose that $C_1\sim C_2$ and $C_2\sim C_3$. Let $a_1,a_2,a_3$ be the maximal elements of $C_1,C_2,C_3$, respectively. Then $C_1\smallsetminus\brac{a_1}$ and $C_2\smallsetminus\brac{a_2}$ are isomorphic and so are $C_2\smallsetminus\brac{a_2}$ and $C_3\smallsetminus\brac{a_3}$. Since being isomorphic is an equivalence relation, we conclude that $C_1\smallsetminus\brac{a_1}$ and $C_3\smallsetminus\brac{a_3}$ are isomorphic. Hence, $C_1\sim C_3$.
\end{proof}

\begin{thm}       
\label{thm23}
The number of producers of a causet $y$ is the number of inequivalent maximal chains in $y$.
\end{thm}
\begin{proof}
If $x$ produces $y$, then $y=x\uparrow a$ and $a$ is the maximal element of at least one maximal chain in $y$. Conversely, if $C$ is a maximal chain in $y$ with maximal element $a$, then $y\smallsetminus\brac{a}$ produces $y$ with
$y=\paren{y\smallsetminus\brac{a}}\uparrow a$. If $C_1$ is another maximal chain that is inequivalent to $C$ and $b$ is the maximal element  of $C_1$, then $y\smallsetminus\brac{a}$ and $y\smallsetminus\brac{b}$ are nonisomorphic producers of $y$. Hence, the correspondence between producers of $y$ and inequivalent maximal chains in $y$ is a bijection.
\end{proof}

\begin{exam}{1}
The following causet has two inequivalent maximal chains.
\end{exam}

\begin{center}
\includegraphics[scale=1.25]{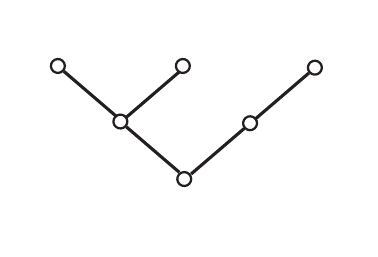}
\end{center}

\noindent and the two producers
\begin{center}
\includegraphics[scale=1.25]{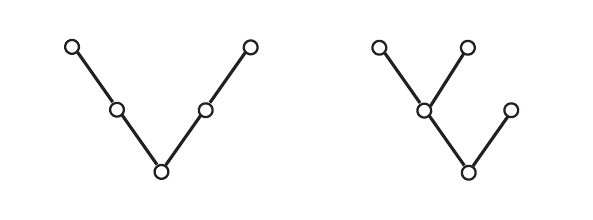}
\end{center}

Notice that equivalent maximal chains may appear to be quite different. In particular, if two (or more) maximal chains have the same maximal element , they are equivalent.

\begin{exam}{2}
The three maximal chains in the following causet are equivalent, so this causet has just one producer.
\end{exam}

\begin{center}
\includegraphics[scale=1.25]{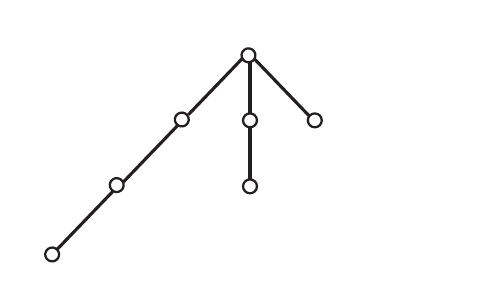}
\end{center}

A \textit{path} in $\pscript$ is a sequence (string) $\omega _1\omega _2\ldots$, where $\omega _i\in\pscript _i$ and
$\omega _i\to\omega _{i+1}$, $i=1,2,\ldots\,$. An $n$-\textit{path} in $\pscript$ is a finite string
$\omega _1\omega _2\ldots\omega _n$, where again $\omega _i\in\pscript _i$ and $\omega _i\to\omega _{i+1}$.
We denote the set of paths by $\Omega$ and the set of $n$-paths by $\Omega _n$. If
$\omega =\omega _1\omega _2\ldots\omega _n\in\Omega _n$ we define $(\omega\rightarrow )\subseteq\Omega _{n+1}$ by
\begin{equation*}
(\omega\rightarrow )=\brac{\omega _1\omega _2\ldots\omega _n\omega _{n+1}\colon\omega _n\to\omega _{n+1}}
\end{equation*}
Thus, $(\omega\rightarrow )$ is the set of one-step continuations of $\omega$. If $A\in\Omega _n$ we define
$(A\rightarrow )\subseteq\Omega _{n+1}$ by
\begin{equation*}
(A\rightarrow )=\bigcup\brac{(\omega\rightarrow )\colon\omega\in A}
\end{equation*}

The set of all paths beginning with $\omega\in\Omega _n$ is called an \textit{elementary cylinder set} and is denoted by
$\rmcyl (\omega )$. If $A\in\Omega _n$, then the \textit{cylinder set} $\rmcyl (A)$ is defined by
\begin{equation*}
\rmcyl (A)=\bigcup\brac{\rmcyl (\omega )\colon\omega\in A}
\end{equation*}
Using the notation
\begin{align*}
\cscript (\Omega _n)&=\brac{\rmcyl (A)\colon A\subseteq\Omega _n}\\
\intertext{we see that}
\cscript (\Omega _1)&\subseteq\cscript (\Omega _2)\subseteq\ldots
\end{align*}
is an increasing sequence of subalgebras of the \textit{cylinder algebra} $\cscript (\Omega )=\cup\,\cscript (\Omega _n)$. Letting $\ascript$ be the $\sigma$-algebra generated by $\cscript (\Omega )$, we have that $(\Omega ,\ascript )$ is a measurable space. For $A\subseteq\Omega$ we define the sets $A^n\subseteq\Omega _n$ by
\begin{equation*}
A^n=\brac{\omega _1\omega _2\ldots\omega _n\colon\omega _1\omega _2\ldots\omega _n\omega _{n+1}\ldots\in A}
\end{equation*}
That is, $A^n$ is the set of $n$-paths that can be continued to a path in $A$. We think of $A^n$ as the $n$-step approximation to $A$. We have that
\begin{equation*}
\rmcyl (A^1)\supseteq\rmcyl (A^2)\supseteq\rmcyl (A^3)\supseteq\ldots\supseteq A
\end{equation*}
so that $A\subseteq\cap\,\rmcyl (A^n)$ but $A\ne\cap\,\rmcyl (A^n)$ in general even if $A\in\ascript$.

\section{Quantum Sequential Growth Processes} 
Let $H_n=L_2(\Omega _n)$ be the $n$-\textit{path Hilbert space} $\complex ^{\Omega _n}$ with the usual inner product
\begin{equation*}
\elbows{f,g}=\sum\brac{\overline{f(\omega )}g(\omega )\colon\omega\in\Omega _n}
\end{equation*}
For $A\subseteq\Omega _n$, the characteristic function $\chi _A\in H_n$ has norm $\|\chi _A\|=\sqrt{\ab{A}\,}$. In particular, $1_n=\chi _{\Omega _n}$ satisfies $\|1_n\|=\sqrt{\ab{\Omega _n}\,}$. A positive operator $\rho$ on $H_n$ that satisfies
$\elbows{\rho 1_n,1_n}=1$ is called a \textit{probability operator}. Corresponding to a probability operator $\rho$ we define the \textit{decoherence functional} $D_\rho\colon 2^{\Omega _n}\times 2^{\Omega _n}\to\complex$ by
\begin{equation*}
D_\rho (A,B)=\elbows{\rho\chi _B,\chi _A}
\end{equation*}
Notice that $D_\rho$ has the usual properties of a decoherence functional. That is, $D_\rho (\Omega _n,\Omega _n)=1$,
$D_\rho (A,B)=\overline{D_\rho (B,A)}$, $A\mapsto D_\rho (A,B)$ is a complex measure on $2^{\Omega _n}$ for every
$B\subseteq\Omega _n$ and if $A_i\subseteq\Omega _n$, $i=1,2,\ldots ,r$, then the $r\times r$ matrix with components
$D_\rho (A_i,A_j)$ is positive semidefinite \cite{hen09,sor94,sor03,sur11}.

We interpret $D_\rho (A,B)$ as a measure of the interference between the events $A,B$ when the system is described by
$\rho$. We also define the $q$-\textit{measure} $\mu _\rho\colon 2^{\Omega _n}\to\real ^+$ by $\mu _\rho (A)=D_\rho (A,A)$ and interpret $\mu _\rho (A)$ as the quantum propensity of the event $A\subseteq\Omega _n$. In general, $\mu _\rho$ is not additive on $2^{\Omega _n}$ so $\mu _\rho$ is not a measure. However, $\mu _\rho$ is \textit{grade}-2 \textit{additive}
\cite{sor94,sor03} in the sense that if $A,B,C\in 2^{\Omega _n}$ are mutually disjoint, then
\begin{equation*}
\mu _\rho (A\cup B\cup C)=\mu _\rho (A\cup B)+\mu _\rho (A\cup C)+\mu _\rho (B\cup C)
-\mu _\rho (A)-\mu _\rho (B)-\mu _\rho (C)
\end{equation*}
Let $\rho _n$ be a probability operator on $H_n$, $n=1,2,\ldots\,$. We say that the sequence $\brac{\rho _n}$ is
\textit{consistent} \cite{gud13} if
\begin{equation*}
D_{\rho _{n+1}}(A\rightarrow, B\rightarrow )=D_{\rho _n}(A,B)
\end{equation*}
for every $A,B\subseteq\Omega _n$. We call a consistent sequence $\brac{\rho _n}$ a
\textit{quantum sequential growth process} (QSGP). Let $\brac{\rho _n}$ be a QSGP and denote the corresponding
$q$-measures by $\mu _n$. A set $A\in\ascript$ is called \textit{suitable} if $\lim\mu _n(A^n)$ exists (and is finite) in which case we define $\mu (A)=\lim\mu _n(A^n)$. We denote the collection of suitable sets by $\sscript (\Omega )$. Of course,
$\emptyset ,\Omega\in\sscript (\Omega )$ with $\mu (\emptyset )=0$, $\mu (\Omega )=1$. If $A\in\cscript (\Omega )$, then
$A=\rmcyl (B)$ where $B\subseteq\Omega _n$ for some $n\in\positive$. Since $A^n=B$, $A^{n+1}=B\rightarrow$,
$A^{n+2}=(B\rightarrow )\rightarrow ,\ldots$, it follows from consistency that $\lim\mu _n(A^n)=\mu _n(B)$. Hence,
$A\in\sscript (\Omega )$ and $\mu (A)=\mu _n(B)$. It follows that
$\cscript (\Omega )\subseteq\sscript (\Omega )\subseteq\ascript$ and it can be shown that the inclusions are proper, in general. In a certain sense, $\mu$ is a quantum measure on $\sscript (\Omega )$ that extends the $q$-measures $\mu _n$.

There are physically relevant sets in $\ascript$ that are not in $\cscript (\Omega )$. In this case it is important to know whether such a set $A$ is in $\sscript (\Omega )$ and to find $\mu (A)$. For example, if $\omega\in\Omega$ then
\begin{equation*}
\brac{\omega}=\bigcap _{n=1}^\infty\brac{\omega}^n\in\ascript
\end{equation*}
but $\brac{\omega}\notin\cscript (\Omega )$. It is of interest whether $\brac{\omega}\in\sscript (\Omega)$ and if so, to 
find $\mu\paren{\brac{\omega}}$. As another example, the complement $\brac{\omega}'\notin\cscript (\Omega )$. Even if
$\brac{\omega}\in\sscript (\Omega )$, since $\mu _n(A')\ne 1-\mu _n(A)$ for $A\subseteq\Omega _n$, it does not immediately follow that $\brac{\omega}'\in\sscript (\Omega )$. For this reason we would have to treat $\brac{\omega}'$ as a separate case.

\section{Discrete Einstein Equations} 
Let $\brac{\rho _n}$ be a QSGP with corresponding decoherence matrices
\begin{equation*}
D_n(\omega ,\omega ')=D_n\paren{\brac{\omega},\brac{\omega '}}
\end{equation*}
$\omega ,\omega '\in\Omega _n$. If $\omega =\omega _1\omega _2\ldots\omega _n\in\Omega _n$ and $\omega _i=x$ for some $i$, then $\omega$ \textit{contains} $x$. Notice that $\omega$ contains $x$ if and only if $\omega _{\ab{x}}=x$. For
$x,y\in\pscript$ with $\ab{x},\ab{y}\le n$ we define
\begin{equation*}
D(x,y)=\sum\brac{D_n(\omega ,\omega ')\colon\omega\text{ contains } x,\omega '\text{ contains }y}
\end{equation*}
If $A_x\subseteq\Omega _n$ is the set $A_x=\brac{\omega\in\Omega _n\colon\omega _{\ab{x}}=x}$ and similarly for
$A_y\subseteq\Omega _n$, then $D(x,y)=D_n(A_x,A_y)$. Due to the consistency of $\brac{\rho _n}$, $D(x,y)$ is independent of $n$ if $\ab{x},\ab{y}\le n$. Also, $D(x,y)$ for $\ab{x},\ab{y}\le n$, are the components of a positive semidefinite matrix. We think of
$\pscript$ as a discrete analogue of a differentiable manifold and $D(x,y)$ as a discrete analogue of a metric tensor.

Let $K=L_2(\pscript )$ be the Hilbert space of square summable complex-valued functions on $\pscript$ with the standard inner product
\begin{equation*}
\elbows{f,g}=\sum _{x\in\pscript}\overline{f(x)}g(x)
\end{equation*}
Let $L=K\otimes K$ which we can identify with the space of square summable complex-valued functions on
$\pscript\times\pscript$. For $\omega ,\omega '\in\Omega$ we define the \textit{covariant bidifference operator}
$\nabla _{\omega ,\omega '}$ on $L$ \cite{gud122} by
\begin{equation*}
\nabla _{\omega ,\omega '}f(x,y)
=\sqbrac{D(\omega _{\ab{x}-1},\omega '_{\ab{y}-1})f(x,y)-D(x,y)f(\omega _{\ab{x}-1},\omega '_{\ab{y}-1})}
  \delta _{x,\omega _{\ab{x}}}\delta _{y,\omega '_{\ab{y}}}
\end{equation*}
In general, $\nabla _{\omega ,\omega '}$ may be unbounded but it is densely defined. The covariant designation stems from the fact that $\nabla _{\omega ,\omega '}D(x,y)=0$ for every $x,y\in\pscript$, $\omega ,\omega '\in\Omega$.

In analogy to the curvature tensor on a manifold, we define the \textit{discrete curvature operator} $\rscript _{\omega ,\omega '}$ on $L$ by
\begin{equation*}
\rscript _{\omega ,\omega '}=\nabla _{\omega ,\omega '}-\nabla _{\omega ',\omega}
\end{equation*}
We also define the \textit{discrete metric operator} $\dscript _{\omega ,\omega '}$ on $L$ by
\begin{align*}
\dscript&_{\omega ,\omega '}f(x,y)\\
&=D(x,y)
\left[f(\omega '_{\ab{x}-1},\omega _{\ab{y}-1})\delta _{x,\omega '_{\ab{x}}}\delta _{y,\omega _{\ab{y}}}\right.
  \left.-f(\omega _{\ab{x}-1},\omega '_{\ab{y}-1})\delta _{x,\omega _{\ab{x}}}\delta _{y,\omega '_{\ab{y}}}\right]
\end{align*}
and the \textit{discrete mass-energy} operator $\tscript _{\omega ,\omega '}$ on $L$ by
\begin{align*}
\tscript&_{\omega ,\omega '}f(x,y)\\
&=\left[D(\omega _{\ab{x}-1},\omega '_{\ab{y}-1})\delta _{x,\omega _{\ab{x}}}\delta _{y,\omega '_{\ab{y}}}\right.
  \left.-D(\omega '_{\ab{x}-1},\omega _{\ab{y}-1})\delta _{x,\omega '_{\ab{x}}}\delta _{y,\omega _{\ab{y}}}\right]f(x,y)
\end{align*}
It is not hard to show that
\begin{equation}         
\label{eq41}
\rscript _{\omega ,\omega '}=\dscript _{\omega ,\omega '}+\tscript _{\omega ,\omega '}
\end{equation}
We call $\eqref{eq41}$ the \textit{discrete Einstein equations} \cite{gud122}.

If we can find $D(x,y)$ such that the classical Einstein equations are an approximation to \eqref{eq41} then it would give information about $D(x,y)$. Moreover, an important problem in discrete quantum gravity theory is to test whether general relativity is a close approximation to the theory. Whether Einstein's classical equations are an approximation to \eqref{eq41} would provide such a test.

We can obtain a simplification by considering the \textit{contracted discrete operators}
$\rscripthat _{\omega ,\omega '}$, $\dscripthat _{\omega ,\omega '}$, $\tscripthat _{\omega ,\omega '}$, from their domains in
$L$ into $K$, respectively, given by $\rscripthat _{\omega ,\omega '}f(x)=\rscript _{\omega ,\omega '}f(x,x)$,
$\dscripthat _{\omega ,\omega '}f(x)=\dscript _{\omega ,\omega '}f(x,x)$,
$\tscripthat _{\omega ,\omega '}f(x)=\tscript _{\omega ,\omega '}f(x,x)$.
We then have the \textit{contracted discrete Einstein equations}
\begin{equation*}
\rscripthat _{\omega ,\omega '}=\dscripthat _{\omega ,\omega '}+\tscripthat _{\omega ',\omega}
\end{equation*}
Defining $\mu (x)=\dscript (x,x)$ for every $x\in\pscript$ we have
\begin{align*}
\dscripthat _{\omega ,\omega '}f(x)&=\mu (x)
\sqbrac{f(\omega '_{\ab{x}-1},\omega _{\ab{x}-1})-f(\omega _{\ab{x}-1},\omega '_{\ab{x}-1}}
\delta _{x,\omega '_{\ab{x}}}\delta _{x,\omega _{\ab{x}}}\\
\tscripthat _{\omega ,\omega '}f(x)&=\sqbrac{2i\itim D(\omega _{\ab{x}-1},\omega '_{\ab{x}-1}}
\delta _{x,\omega _{\ab{x}}}\delta _{x,\omega '_{\ab{x}}}f(x,x)
\end{align*}

We can obtain a better understanding of these operators by seeing how they are realized on basis vectors. For $x\in\pscript$ define the unit vector $e_x=\chi _{\brac{x}}$ in $K$. Then $\brac{e_x\colon x\in\pscript}$ forms an orthonormal basis for $K$ and $S=\brac{e_x\otimes e_y\colon x,y\in\pscript}$ forms an orthonormal basis for $L$. The subspace spanned by $S$ is a dense subspace of $L$ on which the operators are defined. Let $\omega ,\omega '\in\Omega$ and $x,y\in\pscript$. We call
$(x,y)$ an $(\omega ,\omega ')$ \textit{pair} if $\omega$ contains $x$ and $\omega '$ contains $y$. If $(x,y)$ is both an
$(\omega ,\omega ')$ pair and an $(\omega ',\omega )$ pair, then $(x,y)$ is an $(\omega ,\omega ')-(\omega ',\omega )$
\textit{pair}. Then all our operators vanish except possibly at points $(x,y)$ that are $(\omega ,\omega ')$ or
$(\omega ',\omega )$ pairs. We conclude that these operators are local in the sense that they vanish except along the two paths $\omega ,\omega '$.

\begin{thm}       
\label{thm41}
{\rm (a)}\enspace If $(x,y)$ is an $(\omega ,\omega ')$ pair but not an $(\omega ',\omega )$ pair, then
\begin{align*}
\dscript _{\omega ,\omega '}e_x\otimes e_y&=-D(\omega _{\ab{x}+1},\omega '_{\ab{y}+1})
  e_{\omega _{\ab{x}+1}}\otimes e_{\omega '_{\ab{y}+1}}\\
\tscript _{\omega ,\omega '}e_x\otimes e_y&=D(\omega _{\ab{x}-1},\omega '_{\ab{y}-1})e_x\otimes e_y
\end{align*}
{\rm (b)}\enspace If $(x,y)$ is an $(\omega ',\omega )$ pair but not an $(\omega ,\omega ')$ pair, then
\begin{align*}
\dscript _{\omega ,\omega '}e_x\otimes e_y&=D(\omega '_{\ab{x}+1},\omega _{\ab{y}+1})
  e_{\omega '_{\ab{x}+1}}\otimes e_{\omega _{\ab{y}+1}}\\
\tscript _{\omega ,\omega '}e_x\otimes e_y&=-D(\omega '_{\ab{x}-1},\omega _{\ab{y}-1})e_x\otimes e_y
\end{align*}
{\rm (c)}\enspace If $(x,y)$ is an $(\omega ,\omega ')-(\omega ',\omega )$ pair, then
\begin{align*}
\dscript _{\omega ,\omega '}e_x\otimes e_y&=D(\omega '_{\ab{x}+1},\omega _{\ab{y}+1})
  e_{\omega '_{\ab{x}+1}}\otimes e_{\omega _{\ab{y}+1}}\\
 &\quad -D(\omega _{\ab{x}+1},\omega '_{\ab{y}+1})
  e_{\omega _{\ab{x}+1}}\otimes e_{\omega '_{\ab{y}+1}}\\
\tscript _{\omega ,\omega '}e_x\otimes e_y
  &=\sqbrac{D(\omega _{\ab{x}-1},\omega '_{\ab{y}-1})-D(\omega '_{\ab{x}-1},\omega _{\ab{y}-1})}e_x\otimes e_y
\end{align*}
\end{thm}
\begin{proof}
We shall prove Part~(a) and the other parts are similar. If $(x,y)$ is an $(\omega ,\omega ')$ pair but not an
$(\omega ',\omega )$ pair, then
\begin{align}         
\label{eq42}
(\dscript _{\omega ,\omega '}&e_x\otimes e_y)(u,v)\notag\\
  &=D(u,v)\left[e_x\otimes e_y(\omega '_{\ab{u}-1},\omega _{\ab{v}-1})
  \delta _{u,\omega '_{\ab{u}}}\delta _{v,\omega _{\ab{v}}}\right.\notag\\
 &\quad\left.-e_x\otimes e_y(\omega _{\ab{u}-1},\omega '_{\ab{v}-1})
 \delta _{u,\omega _{\ab{u}}}\delta _{v,\omega '_{\ab{v}}}\right]
\end{align}
The right side of \eqref{eq42} vanishes unless $\omega '_{\ab{u}-1}=x$ and $\omega _{\ab{v}-1}=y$ or
$\omega _{\ab{u}-1}=x$ and $\omega '_{\ab{v}-1}=y$. Since $(x,y)$ is not an $(\omega ',\omega )$ pair, the second alternative applies. This term does not vanish only if
$\omega _{\ab{u}-1}=\omega _{\ab{x}}$ and $\omega '_{\ab{v}-1}=\omega _{\ab{y}}$ so we have
\begin{align*}
u&=\omega _{\ab{u}}=\omega _{\ab{x}+1}\\
\intertext{and}
v&=\omega '_{\ab{v}}=\omega '_{\ab{y}+1}
\end{align*}
In this case we have
\begin{equation*}
\dscript _{\omega ,\omega '}e_x\otimes e_y
  =-D(\omega _{\ab{x}+1},\omega '_{\ab{y}+1})e_{\omega _{\ab{x}+1}}\otimes e_{\omega '_{\ab{y}+1}}
\end{equation*}
In a similar way, when $(x,y)$ is an $(\omega ,\omega ')$ pair but not an $(\omega ',\omega )$ pair we have
\begin{equation}         
\label{eq43}
(\tscript _{\omega ,\omega '}e_x\otimes e_y)(u,v)
  =\sqbrac{D(\omega _{\ab{u}-1},\omega '_{\ab{v}-1})\delta _{u,\omega _{\ab{u}}}\delta _{v,\omega '_{\ab{v}}}}
  e_x\otimes e_y(u,v)
\end{equation}
The right side of \eqref{eq43} vanishes unless $u=x$ and $v=y$. In this case we have
\begin{equation*}
\tscript _{\omega ,\omega '}e_x\otimes e_y
  =-D(\omega _{\ab{x}-1},\omega '_{\ab{y}-1})e_x\otimes e_y\qedhere
\end{equation*}
\end{proof}

An interesting special case is when $(x,x)$ is an $(\omega ,\omega ')$ pair. We then have
\begin{align*}
\dscript _{\omega ,\omega '}e_x\otimes e_x&=D(\omega '_{\ab{x}+1},\omega _{\ab{x}+1})
  e_{\omega '_{\ab{x}+1}}\otimes e_{\omega _{\ab{x}+1}}\\
  &\quad -D(\omega _{\ab{x}+1},\omega '_{\ab{x}+1})e_{\omega _{\ab{x}+1}}\otimes e_{\omega '_{\ab{x}+1}}
  \end{align*}
Notice that this is an entanglement of $e_{\omega _{\ab{x}+1}}$ and $e_{\omega '_{\ab{x}+1}}$. We also have
\begin{equation*}
\tscript _{\omega ,\omega '}e_x\otimes e_y
  =2i\itim D(\omega _{\ab{x}-1},\omega '_{\ab{x}-1})e_x\otimes e_x
\end{equation*}

For the contracted operators we have that $\dscripthat _{\omega ,\omega '}e_x\otimes e_y=0$ except for the cases
\begin{align*}
\dscripthat _{\omega ,\omega '}e_{\omega '_{\ab{x}-1}}\otimes e_{\omega _{\ab{x}-1}}&=\mu (x)e_x\\
\dscript _{\omega ,\omega '}e_{\omega _{\ab{x}-1}}\otimes e_{\omega '_{\ab{x}-1}}&=-\mu (x)e_x\\
\end{align*}
when $\omega _{\ab{x}}=\omega '_{\ab{x}}=x$ and $\omega _{\ab{x}-1}\ne\omega '_{\ab{x}-1}$. Moreover,
$\tscripthat _{\omega ,\omega '}e_x\otimes e_y=0$ except for the case
\begin{equation*}
\tscripthat _{\omega ,\omega '}e_x\otimes e_x
  =2i\itim D(\omega _{\ab{x}-1},\omega '_{\ab{x}-1})e_x
\end{equation*}
when $\omega _{\ab{x}}=\omega '_{\ab{x}}=x$.

It is clear that $\tscript _{\omega ,\omega '}$ is a diagonal operator and hence $\tscript _{\omega ,\omega '}$ is a normal operator. We shall show shortly that $\dscript _{\omega ,\omega '}$ is not normal. However, from Theorem~\ref{thm41} we see that $\dscript _{\omega ,\omega '}$ is a type of shift operator. More physically, we see that $\dscript _{\omega ,\omega '}$ can be thought of as a creation operator because it takes $e_x\otimes e_y$ corresponding to $(\ab{x},\ab{y})$ ``particles'' to a scalar multiple of $e_{\omega _{\ab{x}+1}}\otimes e_{\omega _{\ab{y}+1}}$ corresponding to $(\ab{x}+1,\ab{y}+1)$ ``particles''
(Theorem~\ref{thm41}(a),(b)). The next result shows that the adjoint $\dscript _{\omega ,\omega '}^*$ can be thought of as an annihilation operator. We conclude that this formalism not only gives a discrete version of general relativity, there is also emerging a discrete analogue of quantum field theory.

\begin{thm}       
\label{thm42}
{\rm (a)}\enspace If $(x,y)$ is an $(\omega ,\omega ')$ pair but not an $(\omega ',\omega )$ pair, then
\begin{equation*}
\dscript _{\omega ,\omega '}^*e_x\otimes e_y=-\dbar (x,y)e_{\omega _{\ab{x}-1}}\otimes e_{\omega '_{\ab{y}-1}}
\end{equation*}
{\rm (b)}\enspace If $(x,y)$ is an $(\omega ',\omega )$ pair but not an $(\omega ,\omega ')$ pair, then
\begin{equation*}
\dscript _{\omega ,\omega '}^*e_x\otimes e_y=\dbar (x,y)e_{\omega '_{\ab{x}-1}}\otimes e_{\omega _{\ab{y}-1}}
\end{equation*}
{\rm (c)}\enspace If $(x,y)$ is an $(\omega ,\omega ')-(\omega ',\omega )$ pair, then
\begin{equation*}
\dscript _{\omega ,\omega '}^*e_x\otimes e_y=\dbar (x,y)
\sqbrac{e_{\omega '_{\ab{x}-1}}\otimes e_{\omega _{\ab{y}-1}}-e_{\omega _{\ab{x}-1}}\otimes e_{\omega '_{\ab{y}-1}}}
\end{equation*}
\end{thm}
\begin{proof}
We shall prove Part~(a) and the other parts are similar. Let $T$ be the operator on $L$ satisfying
\begin{equation*}
Te_x\otimes e_y
  =-\dbar (x,y)e_{\omega _{\ab{x}-1}}\otimes e_{\omega '_{\ab{y}-1}}
\end{equation*}
We then have
\begin{align*}
\elbows{Te_x\otimes e_y,e_u\otimes e_v}
  &=-D(x,y)\elbows{e_{\omega _{\ab{x}-1}}\otimes e_{\omega '_{\ab{y}-1}},e_u\otimes e_v}\\
  &=-D(x,y)\delta _{\omega _{\ab{x}-1},u}\delta _{\omega '_{\ab{y}-1},v}
\end{align*}
Now $\omega _{\ab{x}-1}=u$ if and only if $\omega _{\ab{u}+1}=x$ and $\omega '_{\ab{y}-1}=v$ if and only if
$\omega _{\ab{v}+1}=y$. We conclude that
\begin{align*}
\elbows{e_x\otimes e_y,\dscript _{\omega ,\omega '}e_u\otimes e_v}
  &=-D(\omega _{\ab{u}+1},\omega '_{\ab{v}+1})
  \elbows{e_x\otimes e_y,e_{\omega _{\ab{u}+1}}\otimes e_{\omega '_{\ab{v}+1}}}\\
  &=-D(\omega _{\ab{u}+1},\omega '_{\ab{v}+1})\delta _{x,\omega _{\ab{u}+1}}\delta _{y,\omega '_{\ab{v}+1}}\\
  &=-D(x,y)\delta _{\omega _{\ab{x}-1},u}\delta _{\omega '_{\ab{y}-1},v}\\
  &=\elbows{Te_x\otimes e_y,e_u\otimes e_v}
\end{align*}
Hence,
\begin{equation*}
Te_x\otimes e_y=\dscript _{\omega ,\omega '}^*e_x\otimes e_y
\end{equation*}
and the result holds.
\end{proof}

In the case when Theorems~\ref{thm41}(a) and \ref{thm42}(a) are both applicable we have that
\begin{align*}
\dscript _{\omega ,\omega '}\dscript _{\omega ,\omega '}^*e_x\otimes e_y&=\ab{D(x,y)}^2e_x\otimes e_y\\
\intertext{and}
\dscript _{\omega ,\omega '}^*,\dscript _{\omega ,\omega '}e_x\otimes e_y
  &=\ab{D(\omega _{\ab{x}+1},\omega '_{\ab{y}+1})}^2e_x\otimes e_y
\end{align*}
so $\dscript _{\omega ,\omega '}$ and $\dscript _{\omega ,\omega '}^*$ do not commute in general. Again,
$\dscript _{\omega ,\omega '}$ and $\tscript _{\omega ,\omega '}$ do not commute because
\begin{align*}
\dscript _{\omega ,\omega '},\tscript _{\omega ,\omega '}e_x\otimes e_y
  &=-D(\omega _{\ab{x}+1},\omega '_{\ab{y}+1})D(\omega _{\ab{x}-1},\omega '_{\ab{y}-1})
  e_{\omega _{\ab{x}+1}}\otimes e_{\omega '_{\ab{y}+1}}\\
\intertext{while}
\tscript _{\omega ,\omega '}\dscript _{\omega ,\omega '}e_x\otimes e_y
  &=-D(\omega _{\ab{x}+1},\omega '_{\ab{y}+1})D(x,y)e_{\omega _{\ab{x}+1}}\otimes e_{\omega '_{\ab{y}+1}}
\end{align*}

\section{Amplitude Processes} 
Various ways of constructing QSGP have been considered \cite{gud112,gud13}. Here we introduce a method called an amplitude process (AP). Although a QSGP need not be generated by an AP, we shall characterize those that are so generated. In the next section we shall present a concrete realization of an AP in terms of a natural quantum action

A \textit{transition amplitude} is a map $\atilde\colon\pscript\times\pscript\to\complex$ such that $\atilde (x,y)=0$ if $x\not\to y$ and $\sum _y\atilde (x,y)=1$ for all $x\in\pscript$. This is similar to a Markov chain except $\atilde (x,y)$ may be complex. The
\textit{amplitude process} (AP) corresponding to $\atilde$ is given by the maps $a_n\colon\Omega _n\to\complex$ where
\begin{equation*}
a_n(\omega _1\omega _2\cdots\omega _n)
  =\atilde (\omega _1,\omega _2)\atilde (\omega _2,\omega _3)\cdots\atilde (\omega _{n-1},\omega _n)
\end{equation*}
We can consider $a_n$ to be a vector in $H_n=L_2(\Omega _n)$. Notice that
\begin{align*}
\elbows{1_n,a_n}&=\sum _{\omega\in\Omega _n}a_n(\omega )=1\\
\intertext{and}
 \|a_n\|&=\sqrt{\sum _{\omega\in\Omega _n}\ab{a_n(\omega )}^2\,}
\end{align*}
Define the rank-1 positive operator $\rho _n=\ket{a_n}\bra{a_n}$ on $H_n$. The norm of $\rho _n$ is
\begin{equation*}
\|\rho _n\|=\|a_n\|^2=\sum _{\omega\in\Omega _n}\ab{a_n(\omega )}^2
\end{equation*}
Since
\begin{equation*}
\elbows{e_n1_n,1_n}=\ab{\elbows{1_n,a_n}}^2=1
\end{equation*}
we conclude that $\rho _n$ is a probability operator.

The corresponding decoherence functional becomes
\begin{align*}
D_n(A,B)&=\elbows{\rho _n\chi _B,\chi _A}=\elbows{\chi _B,a_n}\elbows{a_n,\chi _A}\\
  &=\sum _{\omega\in A}\overline{a_n}(\omega )\sum _{\omega\in B}a_n(\omega )
\end{align*}
In particular, for $\omega ,\omega '\in\Omega _n$, $D_n(\omega ,\omega ')=\overline{a_n}(\omega )a_n(\omega ')$ are the matrix elements of $\rho _n$ in the standard basis. Defining the complex-valued measure $\nu _n$ on $2^{\Omega _n}$ by
$\nu _n(A)=\sum _{\omega\in A}a_n(\omega )$ we see that
\begin{equation*}
D_n(A,B)=\overline{\nu _n(A)}\nu _n(B)
\end{equation*}
The $q$-measure $\mu _n\colon 2^{\Omega _n}\to\real ^2$ becomes
\begin{equation*}
\mu _n(A)=D_n(A,A)=\ab{\nu (A)}^2=\ab{\sum _{\omega\in A}a_n(\omega )}^2
\end{equation*}
In particular, $\mu _n(\omega )=\ab{a_n(\omega )}^2$ for every $\omega\in\Omega _n$ and $\mu _n(\Omega _n)=1$.

\begin{thm}       
\label{thm51}
The sequence of operators $\rho _n=\ket{a_n}\bra{a_n}$ forms a QSGP.
\end{thm}
\begin{proof}
We only need to show that $\brac{\rho _n}$ is a consistent sequence. Using the notation
$\omega\omega _{n+1}=\omega _1\omega _2\ldots\omega _n\omega _{n+1}$, for $A,B\in 2^{\Omega _n}$ we have
\begin{align*}
D_{n+1}(A,B)&=\sum _{\omega\in A\rightarrow}\overline{a_n(\omega )}\sum _{\omega\in B\rightarrow}a_n(\omega )\\
&=\sum\brac{\overline{a_n}(\omega\omega _{n+1})\colon\omega\in A,\omega _n\to\omega _{n+1}}\\
  &\quad\times\sum\brac{a_n(\omega\omega _{n+1})\colon\omega\in B,\omega _n\to\omega _{n+1}}\\
  &=\sum\brac{\overline{a_n}(\omega )\overline{\atilde}(\omega _n,\omega _{n+1})\colon
  \omega\in A,\omega _n\to\omega _{n+1}}\\
  &\quad\times\sum\brac{a_n(\omega )\atilde (\omega _n,\omega _{n+1})\colon\omega\in B,\omega _n\to\omega _{n+1}}\\
  &=\sum _{\omega\in A}\overline{a_n}(\omega )\sum _{\omega\in B}a_n(\omega )=D_n(A,B)
\end{align*}
The result follows.
\end{proof}

Theorem~\ref{thm51} shows that an AP $\brac{a_n}$ \textit{generates} a QSGP $\brac{\rho _n}$. Not all QSGP are generated by an AP and the next theorem characterizes those that are so generated.

\begin{thm}       
\label{thm52}
A QSGP $\brac{\rho _n}$ is generated by an AP if and only if $\rho _n=\ket{a_n}\bra{a_n}$, $a_n\in H_n$, are rank-1 operators and if $\omega =\omega _1,\ldots\omega _n$, $\omega '=\omega '_1\ldots\omega ' _n\in\Omega _n$ with
$\omega _n=\omega '_n$ and $\omega _n\to x$, then
\begin{equation}         
\label{eq51}
a_n(\omega ')a_{n+1}(\omega x)=a_n(\omega )a_{n+1}(\omega 'x)
\end{equation}
\end{thm}
\begin{proof}
If $\brac{a_n}$ is an AP, then it is straightforward to show that \eqref{eq51} holds. Conversely, suppose
$\rho _n=\ket{a_n}\bra{a_n}$ is a QSGP and $\rho _n=\ket{a_n}\bra{a_n}$ where \eqref{eq51} holds. We now show that
$\brac{a_n}$ is an AP. If $x,y,\in\pscript$ with $x\to y$, suppose $\ab{x}=n$ and let $\omega\in\Omega _{n+1}$ with
$\omega =\omega _1\cdots\omega _n\omega _{n+1}$ where $\omega _n= x$, $\omega _{n+1}=y$. If $a_n(\omega ')=0$ for every $\omega '\in\Omega _n$ with $\omega '_n=x$ define $\atilde (x,y)=0$. Otherwise, we can assume that
$a_n(\omega _1\cdots\omega _n)\ne 0$ and define
\begin{equation}         
\label{eq52}
\atilde (x,y)=\frac{a_{n+1}(\omega _1\cdots\omega _ny)}{a_n(\omega _1\cdots\omega _n)}
\end{equation}
By \eqref{eq51} this definition is independent of $\omega\in\Omega _n$ with $\omega _n=x$. If $x\not\to y$ we define
$\atilde (x,y)=0$. Since $\brac{\rho _n}$ is consistent we have
\begin{align*}
\sum _y\atilde (x,y)&=\sqbrac{a_n(\omega _1\cdots\omega _n)}^{-1}\sum _ya_n(\omega _1\cdots\omega _ny)\\
  &=\sqbrac{a_n(\omega _1\cdots\omega _n)}^{-1}
  \elbows{\rho _{n+1}\chi _{(\omega\rightarrow )},\chi _{(\Omega _n\rightarrow)}}\\
  &=\sqbrac{a_n(\omega _1\cdots\omega _n)}^{-1}\elbows{\rho _n\chi _{\brac{\omega}},\chi _{\Omega _n}}=1
\end{align*}
so $\atilde$ is a transition amplitude. Applying \eqref{eq52} we have
\begin{align*}
a_n(\omega _1\cdots\omega _n)&=a_{n-1}(\omega _1\cdots\omega _{n-1})\atilde (\omega _{n-1},\omega _n)\\
  &=a_{n_2}(\omega _1\cdots\omega _{n-2})\atilde (\omega _{n-2},\omega _{n-1})\atilde (\omega _{n-1},\omega _n)\\
  &\ \ \vdots\\
  &=\atilde (\omega _1,\omega _2)\atilde (\omega _2,\omega _3)\cdots\atilde (\omega _{n-1},\omega _n)
\end{align*}
It follows that $\brac{a_n}$ is derived from the transition amplitude $\atilde$ so $\brac{a_n}$ is an AP.
\end{proof}

\section{Quantum Action} 
We now present a specific example of an AP that arises from a natural quantum action. For $x\in\pscript$, the \textit{height} $h(x)$ of $x$ is the cardinality of a longest chain in $x$. The \textit{width} $w(x)$ of $x$ is the cardinality of a largest antichain in $x$. Finally, the \textit{area} $A(x)$ of $x$ is given by $A(x)=h(x)w(x)$. Roughly speaking, $h(x)$ corresponds to an internal time in $x$, $w(x)$ corresponds to the mass or energy of $x$ \cite{gud121} and $A(x)$ corresponds to an action for $x$. If
$x\to y$, then $h(y)=h(x)$ or $h(x)+1$ and $w(y)=w(x)$ or $w(x)+1$. In the case $h(y)=h(x)+1$ we call $y$ a
\textit{height offspring} of $x$, in the case $w(y)=w(x)+1$ we call $y$ a \textit{width offspring} of $x$ and if both $h(y)=h(x)$,
$w(y)=w(x)$ hold, we call $y$ a \textit{mild offspring} of $x$. Let $H(x)$, $W(x)$ and $M(x)$ be the sets of height, width and mild offspring of $x$, respectively. It is shown in \cite{gud13} that $H(x)$, $W(x)$, $M(x)$ partition the set $(x\rightarrow )$. It is easy to see that $H(x)\ne\emptyset$, $W(x)\ne\emptyset$ but examples show that $M(x)$ can be empty.

If $x\to y$ we have the following possibilities: $y\in M(x)$ in which case $A(y)-A(x)=0$, $y\in H(x)$ in which case
\begin{equation*}
A(y)-A(x)=\sqbrac{h(x)+1}w(x)-h(x)w(x)=w(x)
\end{equation*}
and $y\in W(x)$ in which case
\begin{equation*}
A(y)-A(x)=h(x)\sqbrac{w(x)+1}-h(x)w(x)=h(x)
\end{equation*}
We now define the transition amplitude $\atilde (x,y)$ in terms of the ``action'' change for $x$ to $y$. We first define the \textit{partition function}
\begin{equation*}
z(x)=\sum _y\brac{m(x\to y)e^{2\pi i\sqbrac{A(y)-A(x)}/\ab{x}}\colon x\to y}
\end{equation*}
If $x\not\to y$ define $\atilde (x,y)=0$. If $x\to y$ and $z(x)=0$ define $\atilde (x,y)=\sqbrac{(x\rightarrow )}^{-1}$ where
$\sqbrac{(x\rightarrow )}$ means the cardinality of $(x\rightarrow )$ including multiplicity. If $x\to y$ and $z(x)\ne 0$ define
\begin{equation*}
\atilde (x,y)=\frac{m(x\to y)}{z(x)}\,e^{2\pi i\sqbrac{A(y)-A(x)}/\ab{x}}
\end{equation*}
As before, we have three possibilities when $z(x)\ne 0$. If $y\in M(x)$, then $\atilde (x,y)=m(x\to y)\sqbrac{z(x)}^{-1}$, if
$y\in H(x)$ then
\begin{equation*}
\atilde (x,y)=\frac{m(x\to y)}{z(x)}\,e^{2\pi iw(x)/\ab{x}}
\end{equation*}
and if $y\in W(x)$, then
\begin{equation*}
\atilde (x,y)=\frac{m(x\to y)}{z(x)}\,e^{2\pi ih(x)/\ab{x}}
\end{equation*}
Notice that for $x\in\pscript _n$ we have
\begin{equation*}
z(x)=\sqbrac{M(x)}+\sqbrac{H(x)}e^{2\pi iw(x)/n}+\sqbrac{W(x)}e^{2\pi ih(x)/n}
\end{equation*}

We now illustrate this theory by checking the first three steps in Figure~1. Since $M(x_1)=\emptyset$, $H(x)=\brac{x_2}$,
$W(x_1)=\brac{x_3}$ and $w(x_1)=h(x_1)=1$ we have $z(x_1)=2e^{2\pi i}=2$. Since $M(x_2)=\emptyset$,
$H(x_2)=\brac{x_4}$, $W(x_2)=\brac{x_5,x_6}$ and $w(x_2)=1$, $h(x_2)=2$ we have
\begin{equation*}
z(x_2)=e^{2\pi i/2}+2e^{2\pi i}=-1+2=1
\end{equation*}
Since $M(x_3)=\emptyset$, $H(x_3)=\brac{x_6,x_7}$, $W(x_3)=\brac{x_8}$, $w(x_3)=2$, $h(x_3)=1$ and
$\sqbrac{H(x_3)}=3$ we have
\begin{equation*}
z(x_3)=3e^{2\pi i}+e^{2\pi i/2}=3-1=2
\end{equation*}
There are six 3-paths in $\Omega _3$: $\gamma _1=x_1x_2x_4$, $\gamma _2=x_1x_2x_5$, $\gamma _3=x_1x_2x_6$,
$\gamma _4=x_1x_3x_6$, $\gamma _5=x_1x_3x_7$ and $\gamma _8=x_1x_3x_8$. The amplitudes of the paths become:
\begin{align*}
a_3(\gamma _1)&=\atilde (x_1,x_2)\atilde (x_2,x_4)=\tfrac{1}{2}e^{2\pi i}e^{\pi i}=-\tfrac{1}{2}\\
a_3(\gamma _2)&=\atilde (x_1,x_2)\atilde (x_2,x_5)=\tfrac{1}{2}e^{2\pi i}e^{2\pi i}=\tfrac{1}{2}\\
a_3(\gamma _3)&=\atilde (x_1,x_2)\atilde (x_2,x_6)=\tfrac{1}{2}e^{2\pi i}e^{2\pi i}=\tfrac{1}{2}\\
a_3(\gamma _4)&=\atilde (x_1,x_3)\atilde (x_3,x_6)=\tfrac{1}{2}e^{2\pi i}e^{2\pi i}=\tfrac{1}{2}\\
a_3(\gamma _5)&=\atilde (x_1,x_3)\atilde (x_3,x_7)=\tfrac{1}{2}e^{2\pi i}\tfrac{1}{2}e^{2\pi i}=\tfrac{1}{4}\\
a_3(\gamma _6)&=\atilde (x_1,x_3)\atilde (x_3,x_8)=\tfrac{1}{2}e^{2\pi i}\tfrac{1}{2}e^{\pi i}=-\tfrac{1}{4}
\end{align*}
Notice that $\sum a(\gamma _i)=1$ as it must. The amplitude decoherence matrix has components
$D_3(\gamma _i,\gamma _j)=\overline{a(\gamma _i)}a(\gamma _j)$, $i,j=1,\ldots ,6$. The $q$-measures
$\mu _3(\gamma _i)=\ab{a(\gamma _i)}^2$, $i=1,\ldots 6$, are given by
\begin{align*}
\mu _3(\gamma _1)&=\mu _3(\gamma _2)=\mu _3(\gamma _3)=\mu _3(\gamma _4)=1/4\\
\mu _3(\gamma _5)&=\mu _3(\gamma _6)=1/16
\end{align*}
Interference effects are evident if we consider the $q$-measures of various sets of sites (causets). For example,
\begin{align*}
\mu _3(\brac{x_6})&=\ab{\tfrac{1}{2}+\tfrac{1}{2}}^2=1\\
\mu _3(\brac{x_4,x_5})&=\ab{-\tfrac{1}{2}+\tfrac{1}{2}}^2=0\\
\mu _3(\brac{x_5,x_6})&=\ab{\tfrac{3}{2}}^2=9/4
\end{align*}
We can easily compute the amplitudes for the sites $x_1,\ldots ,x_8$ to get $a(x_1)=1$ by convention and
$a(x_2)=a(x_3)=a(x_5)=\tfrac{1}{2},a(x_4)=-\tfrac{1}{2}$, $a(x_6)=1$, $a(x_7)=\tfrac{1}{4}$, $a(x_8)=-\tfrac{1}{4}$. The site decoherence matrix is the positive semidefinite matrix with components $D(x_i,x_j)=\overline{a(x_i)}a(x_j)$, $i,j=1,\ldots ,8$.

For this particular AP we conjecture that $\brac{\omega}\in\sscript (\Omega )$ for every $\omega\in\Omega$ and
$\mu\paren{\brac{\omega}}=0$. Moreover, we conjecture that $\brac{\omega}'\in\sscript (\Omega )$ for every
$\omega\in\Omega$ and $\mu\paren{\brac{\omega}'}=1$. Since
\begin{equation*}
\mu _n\paren{\brac{\omega}^n}=\ab{a_n\paren{\brac{\omega}^n}}^2=\paren{\prod _{j=1}^{n-1}\ab{z(\omega _j)}^2}^{-1}
\end{equation*}
$\mu _n\paren{\brac{\omega}}=0$ would follow from
\begin{equation}         
\label{eq61}
\lim _{\ab{x}\to 0}\brac{\ab{z(x)}\colon x\in\pscript}=\infty
\end{equation}
We also conjecture that \eqref{eq61} holds. Notice that $\brac{\omega}'\in\sscript (\Omega )$ would follow from
$\brac{\omega}\in\sscript (\Omega )$ and $\mu\paren{\brac{\omega}}=0$ because
\begin{align*}
\mu _n\paren{{\brac{\omega }'}^{\,n}}&=\ab{a_n\paren{{\brac{\omega}'}^{\,n}}}^2=\ab{1-a_n\paren{\brac{\omega}^n}}^2\\
&=1+\ab{a_n\paren{\brac{\omega}^n}}^2-2\itre a_n\paren{\brac{\omega}^n}
\end{align*}
If $\mu\paren{\brac{\omega}}=0$, then $\lim\ab{a_n\paren{\brac{\omega}^n}^2}=0$ so $\lim a_n\paren{\brac{\omega}^n}=0$ and hence, $\lim\mu _n\paren{{\brac{\omega}'}^{\,n}}=1$.

We can prove our conjectures in two extreme cases. Let $\omega\in\Omega$ be the path
$\omega =\omega _1\omega _2\cdots$ for which $\omega _n$ is a chain, $n=1,2,\ldots\,$. Then
$\sqbrac{H(\omega _n)}=1$, $\sqbrac{W(\omega _n)}=n$ and $\sqbrac{M(\omega _n)}=0$, $n=1=,2,\ldots\,$. Then
$\omega ^n=\brac{\omega}^n=\omega _1\omega _2\ldots\omega _n$ and for $j=1,\ldots ,n$ we have
\begin{equation*}
z(\omega _j)=e^{2\pi i/j}+je^{2\pi ij/j}=e^{2\pi i/j}+j
\end{equation*}
Hence,
\begin{equation*}
\ab{z(\omega _j)}=\ab{e^{2\pi i/j}+j}\ge\ab{j}-\ab{e^{2\pi i/j}}=j-1
\end{equation*}
We then have
\begin{equation*}
\ab{z(\omega ^n)}=\ab{z(\omega _1)\cdots z(\omega _{n-1})}\ge\prod _{j=1}^{n-2}j
\end{equation*}
Hence, $\lim\limits _{n\to\infty}\ab{z(\omega ^n)}=\infty$ so $\brac{\omega}\in\sscript (\Omega )$ and
$\mu\paren{\brac{\omega}}=0$.

The other extreme case is when $\omega =\omega _1\omega _2\cdots$ and each $\omega _n$ is an antichain. Then
$\sqbrac{H(\omega _n)}=2^{n-1}$, $\sqbrac{W(\omega _n)}=1$ and $\sqbrac{M(\omega _n)}=0$, $n=1,2,\ldots\,$. Again,
$\omega ^n=\omega _1\omega _2\cdots\omega _n$ and for $j=1,\ldots ,n$ we obtain
\begin{equation*}
z(\omega _j)=(2^j-1)e^{2\pi i}+e^{2\pi i/j}=2^j-1+e^{2\pi i/j}
\end{equation*}
Hence, $\ab{z(\omega _j)}\ge 2^j-2$ and we obtain
\begin{equation*}
\ab{z(\omega ^n)}=\ab{z(\omega _1)\ldots z(\omega _{n-1})}\ge\prod _{j=1}^{n-2}(2^n-2)
\end{equation*}
Again, we conclude that $\lim\limits _{n\to\infty}\ab{z(\omega ^n)}=\infty$ so $\brac{\omega}\in\sscript (\Omega )$ with
$\mu\paren{\brac{\omega}}=0$. As noted before, in both these extreme cases we have $\brac{\omega}'\in\sscript (\Omega )$ and $\mu\paren{\brac{\omega}'}=1$.

\section{Classical Processes} 
A QSPG $\brac{\rho _n}$ is \textit{classical} if the decoherence matrix
\begin{equation*}
D_n(\omega ,\omega ')=\elbows{\rho _n\chi _{\brac{\omega '}},\chi _{\omega}}
\end{equation*}
is diagonal for all $n$. As an example, let $\atilde\colon\pscript\times\pscript\to\real$ be a real-valued transition amplitude and define the AP $a_n\colon\Omega _n\to\real$ as in Section~5. Defining the diagonal operators
\begin{equation*}
\rho _n(\omega ,\omega ')=a_n(\omega )\delta _{\omega ,\omega '}
\end{equation*}
we conclude that $\brac{\rho _n}$ is a classical QSGP.

\begin{thm}       
\label{thm71}
The following statements are equivalent.
{\rm (a)}\enspace $\brac{\rho _n}$ is classical.
{\rm (b)}\enspace $D_n(A,B)=\mu _n(A\cap B)$ for all $A,B\subseteq\Omega _n$.
{\rm (c)}\enspace $D_n(A,B)=0$ if $A\cap B=\emptyset$.
\end{thm}
\begin{proof}
(a)$\Rightarrow$(b)\enspace If $\rho _n$ is diagonal, then $D_n(\omega ,\omega ')=0$ for $\omega\ne\omega '$. Hence,
\begin{align*}
D_n(A,B)&=\elbows{\rho _n\chi_B,\chi _A}
  =\elbows{\rho _n\sum _{\omega '\in B}\chi _{\brac{\omega '}},\sum _{\omega\in A}\chi _{\brac{\omega}}}\\
  &=\sum _{\omega\in A}\sum _{\omega '\in B}\elbows{\rho _n\chi _{\brac{\omega '}},\chi _{\brac{\omega}}}
  =\sum _{\omega\in A}\sum _{\omega '\in B}D_n(\omega ,\omega ')\\
  &=\sum _{\omega\in A\cap B}D_n(\omega ,\omega )=D_n(A\cap B,A\cap B)=\mu _n(A\cap B)\\
\end{align*}
(b)$\Rightarrow$(c)\enspace If $D_n(A,B)=\mu _n(A\cap B)$ then $A\cap B=\emptyset$
\begin{equation*}
D_n(A,B)=\mu _n(\emptyset )=0
\end{equation*}
(c)$\Rightarrow$(a)\enspace If (c) holds and $\omega\ne\omega '$ then $\brac{\omega}\cap\brac{\omega'}=\emptyset$ so
$D_n(\omega ,\omega ')=0$.
\end{proof}

\begin{cor}       
\label{cor72}
If $\brac{\rho _n}$ is classical, then $\mu _n$ is a measure, $n=1,2,\ldots\,$.
\end{cor}
\begin{proof}
To show that $\mu _n$ is a measure, suppose that $A\cap B=\emptyset$. By Theorem~\ref{thm71} we have
\begin{align*}
\mu _n(A\cup B)&=D_n(A\cup B,A\cup B)=D_n(A,A\cup B)=D_n(B,A\cup B)\\
  &=D_n(A,A)+D_n(B,B)+2\itre D_n(A,B)=\mu _n(A)+\mu _n(B)\quad\qedhere
\end{align*}
\end{proof}

We now give a more general condition. A QSGP $\brac{\rho _n}$ is \textit{semiclassical} if $\itre D_n(\omega ,\omega ')$ is diagonal. The fact that $\mu _n$ is a measure, $n=1,2,\ldots$, does not imply that $\brac{\rho _n}$ is classical. However, we do have the following result whose proof is similar to that of Theorem~\ref{thm71} and Corollary~\ref{cor72}.

\begin{thm}       
\label{thm73}
The following statements are equivalent.
{\rm (a)}\enspace $\brac{\rho _n}$ is semiclassical.
{\rm (b)}\enspace $\itre D_n(A,B)=\mu _n(A\cap B)$.
{\rm (c)}\enspace $\itre D_n(A,B)=0$ if $A\cap B=\emptyset$.
{\rm (d)}\enspace $\mu _n$ is a measure.
\end{thm}

If $\brac{\rho _n}$ is semiclassical, then by Theorem~\ref{thm73}, $\mu _n$ are measures so for every $A\in\ascript$ we have
\begin{equation*}
\mu _{n+1}(A^{n+1})\le\mu _{n+1}(A^n\rightarrow )=\mu _n(A^n)
\end{equation*}
Hence, $\mu _n(A^n)$ is a decreasing sequence so it converges. We conclude that $\sscript (\Omega )=\ascript$. However, the $q$-measure $\mu$ defined earlier need not be a measure on $\ascript$. One reason is that if $A\cap B=\emptyset$, $A,B\in\ascript$ then $A^n\cap B^n\ne\emptyset$, in general. Another way to see this is the following. Define $\nu (A)=\mu _n(A_1)$ for $A=\rmcyl (A_1)$, $A_1\subseteq\Omega _n$. Then by the Hahn extension theorem $\nu$ extends to unique measure on $\ascript$. But $\nu =\mu$ on $\cscript (\Omega )$ so if $\mu$ were a measure, by uniqueness $\nu =\mu$ on $\ascript$. But this is impossible because $\cap A^n\ne A$ in general, so
\begin{equation*}
\mu (A)=\lim\mu _n(A^n)\ne\nu (A)
\end{equation*}

Suppose $\brac{\rho _n}$ is classical for the rest of this section unless specified otherwise. It is of interest to study the operators considered in Section~4 for this case. Since $D_n(\omega ,\omega ')=0$ for $\omega\ne\omega '$ we have for all $x,y\in\pscript$ and $n\ge\ab{x},\ab{y}$ that
\begin{align}         
\label{eq71}
D(x,y)&=\sum\brac{D_n(\omega ,\omega ')\colon\omega _{\ab{x}}=x,\omega '_{\ab{y}}=y}\notag\\
  &=\sum\brac{D_n(\omega ,\omega )\colon\omega _{\ab{x}}=x,\omega _{\ab{y}}=y}\notag\\
  &=\sum\brac{\mu _n(\omega )\colon\omega _{\ab{x}}=x,\omega _{\ab{y}}=y}
\end{align}
Notice that the set
\begin{equation*}
A=\brac{\omega\in\Omega\colon\omega _{\ab{x}}=x,\omega _{\ab{y}}=y}\in\cscript (\Omega )
\end{equation*}
and by \eqref{eq71} $\mu (A)=D(x,y)$. We conclude from \eqref{eq71} that $D(x,y)=0$ except if $x$ and $y$ are comparable. It follows that $\tscripthat _{\omega ,\omega '}=0$ for every $\omega ,\omega '\in\Omega$. Since the mass-energy operators vanish when the process is classical, this indicates that mass-energy is generated by quantum interference. (Is this related to the Higgs boson?) The operators $\dscripthat _{\omega ,\omega '}$ have the same form classically as quantum mechanically.

Now suppose we have a ``flat'' space $\pscript$ so that $\rscripthat _{\omega ,\omega '}=0$ for every
$\omega ,\omega '\in\Omega$. Since we already have that $\tscripthat _{\omega ,\omega '}=0$, it follows that $\dscripthat _{\omega ,\omega '}=0$ for every $\omega ,\omega '\in\Omega$. Now $\dscripthat _{\omega ,\omega '}e_u\otimes e_v=0$ for every $\omega ,\omega '\in\Omega$ unless possibly when $u\ne v$ and $u$ and $v$ have a common offspring $x$
(from which it follows that $\ab{u}=\ab{v}$). In this case we have that
$\dscripthat _{\omega ,\omega '}e_u\otimes e_v=\pm\mu (x)e_x$ so that $\mu (x)=0$. We conclude that
$\rscripthat _{\omega ,\omega '}=0$ for every $\omega ,\omega '\in\Omega$ if and only if $\mu (x)=0$ whenever $x$ has more than one producer.

The next theorem considers the general operators $\dscript _{\omega ,\omega '}$ and $\tscript _{\omega ,\omega '}$
\begin{thm}       
\label{thm74}
The operators $\dscript _{\omega ,\omega '}=0$ for every $\omega ,\omega '\in\Omega$ if and only if $\mu (x)=0$ whenever
$x$ has more than one producer and the operators $\tscript _{\omega ,\omega '}=0$ for every $\omega ,\omega '\in\Omega$ if and only if $\mu (x)=0$ whenever $x$ has more than one offspring.
\end{thm}
\begin{proof}
Suppose $(x,y)$ is an $(\omega ,\omega ')$ pair but not an $(\omega ',\omega )$ pair. Then by Theorem~\ref{thm41}(a) we have that $\dscript _{\omega ,\omega '}e_x\otimes e_y\ne 0$ if and only if $\omega _{\ab{x}+1}=\omega '_{\ab{y}+1}$ and
$x\ne y$ with $\mu (\omega _{\ab{x}+1})\ne 0$. In this case $\omega _{\ab{x}+1}$ has the two producers $x$, $y$ and if
$\dscript _{\omega ,\omega '}e_x\otimes e_y=0$ then $\mu (\omega _{\ab{x}+1})=0$. We get the same result if $(x,y)$ is an
$(\omega ',\omega )$ pair but not an $(\omega ,\omega ')$ pair. Now suppose $(x,y)$ is an
$(\omega ,\omega ')-(\omega ',\omega )$ pair. Then by Theorem~\ref{thm41}(c) we have that
$\dscript _{\omega ,\omega '}e_x\otimes e_y\ne 0$ if and only if $\omega '_{\ab{x}+1}=\omega _{\ab{y}+1}$ or
$\omega _{\ab{x}+1}=\omega '_{\ab{y}+1}$. If either of these hold, then $\ab{x}=\ab{y}$ so $x=y$. We then have by
Theorem~\ref{thm41}(c) that
\begin{align*}
\dscript _{\omega ,\omega '}e_x\otimes e_y
  &=D(\omega '_{\ab{x}+1},\omega _{\ab{x}+1})e_{\omega '_{\ab{x}+1}}\otimes e_{\omega _{\ab{x}+1}}\\
  &\quad -D(\omega _{\ab{\omega}+1},\omega '_{\ab{x}+1})e_{\omega _{\ab{x}+1}}\otimes e_{\omega '_{\ab{x}+1}}\\
  &=0
\end{align*}
We conclude that $\dscript _{\omega ,\omega '}e_x\otimes e_y=0$ which is a contradiction. Hence, in this case
$\dscript _{\omega ,\omega '}e_x\otimes e_y=0$ automatically.

Again, suppose $(x,y)$ is an $(\omega ,\omega ')$ pair but not an $(\omega ',\omega )$ pair. By Theorem~\ref{thm41}(a) we have that $\tscript _{\omega ,\omega '}e_x\otimes e_y\ne 0$ if and only if $\omega _{\ab{x}-1}=\omega '_{\ab{y}-1}$ and
$x\ne y$ with $\mu (\omega _{\ab{x}-1})\ne 0$. In this case $\omega _{\ab{x}-1}$ has the two offspring $x,y$ and if
$\tscript _{\omega ,\omega '}e_x\otimes e_y=0$, then $\mu (\omega _{\ab{x}-1}=0$. The rest of the proof is similar to that in the previous paragraph.
\end{proof}

In very special conditions, it may be possible to arrange things so that $\mu (x)=0$ whenever $x$ has more than one producer in which case $\dscript _{\omega ,\omega '}=0$ for every $\omega ,\omega '\in\Omega$. However, $\mu (x)=0$ whenever $x$ has more than one offspring is impossible. This is because all $x\in\pscript$ have more than one offspring so $\mu (x)=0$ for all $x\in\pscript$. But this is a contradiction because $\mu$ is a probability measure on $\pscript _n$. We conclude that
$\tscript _{\omega ,\omega '}\ne 0$ for every $\omega,\omega '$. It follows that $\rscript _{\omega ,\omega '}\ne 0$ for every
$\omega ,\omega '$. Thus in the classical case $\pscript$ is never ``flat'' in the sense that $\rscript _{\omega ,\omega '}=0$ for all $\omega ,\omega '\in\Omega$ but it may be ``flat'' in the sense that $\rscripthat _{\omega ,\omega '}=0$ for all
$\omega ,\omega '\in\Omega$.

\end{document}